\documentclass[11pt]{article}
\usepackage[utf8]{inputenc}
\usepackage{amssymb,amsmath,amsthm}
\usepackage{mathrsfs, enumitem}
\usepackage{graphicx}
\usepackage{authblk}

\usepackage[style=numeric,sorting=none, backend=bibtex]{biblatex}
\usepackage{hyperref}
\usepackage{xurl}
\hypersetup{breaklinks=true}

\newcommand{\tr}{{\text{tr}}}

\newcommand{\B}{\mathbf}
\newcommand{\BT}[1]{\tilde{\bold{#1}}}

\newtheorem{theorem}{Theorem}
\newtheorem{corollary}{Corollary}[theorem]
\newtheorem{lemma}{Lemma}
\theoremstyle{definition}
\newtheorem*{remark}{Remark}

\numberwithin{equation}{section}

\title{Depolarizing Reference Devices \\in Generalized Probabilistic Theories}
\author[1]{Matthew B. Weiss}
\affil[1]{\emph{QBism Group, University of Massachusetts Boston}, \authorcr \emph{100 Morrissey Boulevard, Boston, MA 02125, USA}}
\date{\today}

\begin{document}
\maketitle

\begin{abstract}
QBism is an interpretation of quantum theory which views quantum mechanics  as standard probability theory supplemented with a few extra normative constraints. The fundamental gambit is to represent states and  measurements, as well as time evolution, with respect to an \emph{informationally complete reference device}. From this point of view, the Born rule appears as a coherence condition on probability assignments across several different experiments which manifests as a deformation of the law of total probability (LTP). In this work, we fully characterize those reference devices for which this deformation takes a ``simplest possible'' (term-wise affine) form.  Working in the framework of generalized probability theories (GPTs), we show that, given any reference measurement,  a set of post-measurement reference states can always be chosen to give its probability rule this very form.  The essential condition is that the corresponding measure-and-prepare channel be depolarizing.  We also relate our construction to Szymusiak and S\l{}omczy\'nski's recently introduced notion of morphophoricity and re-examine critically a matrix-norm-based measure of LTP deformation in light of our results.  What stands out for the QBist project from this analysis is that it is not only the pure form of the Born rule that must be understood normatively, but the constants within it as well. It is they that carry the details of quantum theory.
\end{abstract}

\section{Introduction}

The Born rule is usually understood as a recipe for calculating the probability of a quantum mechanical measurement outcome given two ingredients: an \emph{effect}, that is, a POVM element representing an outcome, and a \emph{state}, a density matrix representing the preparation of a system. Given POVM element $A_i$ and density matrix $\rho$, we have
\begin{align}
P(A_i)&= \tr(A_i \rho).
\end{align}
In contrast, QBists read the Born rule as a normative constraint on an agent's probability assignments for two separate experiments \cite{fuchs2010qbism, fuchs2023qbism, fuchs2009quantumbayesian}. To accomplish this, a QBist specifies an \emph{informationally complete reference device}, consisting of a set of reference effects $\{E_j\}$ and a set of reference states $\{\sigma_j\}$, both of which span the space of operators. Upon obtaining outcome $E_j$, the device prepares the state $\sigma_j$. Informational completeness implies that the probability assignments $P(E_j) = \tr(E_j\rho)$ are completely equivalent to the assignment of a density matrix $\rho$. Indeed, QBists view quantum states doxastically: the formalism centers the fact that quantum states can be viewed quite literally as probability assignments. Not only that but an arbitrary measurement can be fully characterized by an agent's conditional probability assignments $P(A_i|E_j)=\tr(A_i \sigma_j)$. The Born rule then takes the form
\begin{align}
P(A_i)&= \sum_{jk} P(A_i|E_j)\Phi_{jk}P(E_k),
\end{align}
where $\Phi$ is called the \emph{Born matrix}, and depends upon the choice of reference device. From this, it becomes clear how the Born rule can be understood as a coherence condition between probability assignments $P(A_i), P(A_i|E_j)$, and $P(E_i)$. Moreover, if one were to remove $\Phi$ from the equation, one would be left with 
\begin{align}
P(A_i) = \sum_j P(A_i|E_j)P(E_j),
\end{align}
which is just the usual law of total probability (LTP), implying that the Born rule is a kind of \emph{deformation} of the LTP \cite{Ferrie_2011, Ferrie_2008}.

It has long been a source of fascination for QBists that for a particular kind of reference device, the deformation from the LTP takes an alluringly gentle form \cite{irreducible_difference}. If the reference measurement consists of a SIC-POVM in a $d$-dimensional Hilbert space, and the reference states are taken to be proportional to POVM elements, then the Born rule becomes simply
\begin{align}
P(A_i) = \sum_{j=1}^{d^2} P(A_i|E_j)\left[ (d+1)P(E_j) - \frac{1}{d}\right],
\end{align}
an expression which has earned the moniker \emph{Urgleichung}, meaning fundamental equation. That such a form is possible is not at all trivial. A SIC-POVM, or \emph{symmetric informationally complete quantum measurement}, is one whose $d^2$ effects are proportional to pure states forming a regular simplex in quantum state space. Unpacking this, we observe that such a measurement is a MIC, a \emph{minimal} informationally complete measurement, whose elements form a basis: $d^2$ elements is the fewest number of elements an IC measurement can have in quantum theory. Moreover, such a measurement is \emph{unbiased}, meaning that if one assigns the state of maximal ignorance to the input, one ought to assign equal probability to each outcome. The measurement operators are \emph{rank-1}, that is, proportional to projectors onto pure quantum states. The measurement is \emph{equiangular}, which comes from the fact that the effects form a regular simplex. And finally, the overall reference device is a \emph{parallel update} reference device, meaning that the post-measurement reference states are proportional to POVM elements so that one can think about them as having come about by Luders' rule. The equiangularity condition then implies that upon repeating the measurement, one ought to assign probability $\frac{1}{d}$ to obtaining the same outcome, while assigning $\frac{1}{d(d+1)}$ to all other outcomes. Thus SIC reference devices form in a sense the closest analogy to an orthonormal basis measurement while yet being informationally complete.

 In this paper, we work in the broad framework of \emph{generalized probabilistic theories} \cite{muller} and fully characterize those reference devices for which the Born rule can be written in the following general form, which we shall call a \emph{Protourgleichung}:
\begin{align}
P(A_i) = \sum_{j=1}^{d^2} P(A_i|E_j)\Big[ \alpha P(E_j) + (1-\alpha)w_j\Big],
\end{align}
for some constant $\alpha$ and non-negative weights $\sum_i w_i=1$, corresponding to the bias of the reference measurement. We show that in any GPT, if one starts with an arbitrary informationally complete set of reference effects, one can always choose a set of reference states which deliver a Protourgleichung. In fact, the possibility of employing a Protourgleichung is equivalent to the reference device, considered as a measure-and-prepare channel, being depolarizing. Specializing to the case where the depolarizing reference states are taken (essentially) proportional to reference effects, we relate our construction to the recently proposed notion of a \emph{morphophoric measurement} \cite{morphophoric, morpho}, for which the image of the state space within the probability simplex induced by the reference measurement is as nice as possible,---and ultimately to \emph{complex projective 2-designs}, of which SIC-POVMs are a prime example. Finally, we examine critically a measure of LTP deformity introduced in the QBist literature, whereby the difference between the LTP and the Born rule is quantified by $||I-\Phi||$ for some choice of unitarily invariant matrix norm. The Born matrix $\Phi$ is not uniquely defined for overcomplete reference devices, and it turns out that different choices of $\Phi$, each of which acts as an Protourgleichung, may have wildly different values of $||I-\Phi||$; moreover, we show that the $\Phi$ matrix which minimizes $||I-\Phi||$ with respect to the Frobenius norm for unbiased depolarizing reference devices, does not even give rise to a Protourgleichung (except in one special case). The special case turns out to be interesting, however, even as it highlights the limitations of the measure.

The overall message is that Protourgleichung expressions are ubiquitous, both within quantum theory, and beyond. On the one hand, this shows that non-classicality can be almost effortlessly incorporated into probability theory in a remarkably uniform way. On the other hand, the Urgleichung itself remains unique in terms of the constant $\alpha=d+1$ and the number of its reference outcomes, $d^2$. The fact that pure SICs apparently exist in any Hilbert space dimension $d$ is a very remarkable thing indeed \cite{Appleby_2017, Bengtsson_2020}.

\section{Generalized probabilistic theories }
The framework of  \emph{generalized probabilistic theories}, or GPTs, incorporates both classical probability theory and quantum mechanics under its aegis, as well as a vast landscape of alternative theories. The essential idea is to take some convex subset of a vector space as a \emph{state space}; the \emph{effects} are then non-negative linear functionals that act on states to give probabilities $P(E|S)=(E|S)$. The foundational result in the field is known as \emph{Ludwig's embedding theorem} \cite{lami}, which promises that as long as our theory satisfies a few well-motivated prerequisites, we can safely work with a vector space representation. We reproduce a simplified version of it here.
\begin{theorem}
Let $\Sigma$ be the set of all possible states and $\mathcal{E}$ be the set of all possible effects. A theory is a function $\mu: \mathcal{E} \times \Sigma \rightarrow [0, 1]$. Suppose a theory satisfies 3 axioms:
\begin{enumerate}[label=(\roman*)]
	\item If $\forall E \in \mathcal{E} : \mu(E, S_1)=\mu(E, S_2)$, then $S_1=S_2$: if $\forall S \in \Sigma : \mu(E_1, S)=\mu(E_2, S)$, then $E_1=E_2$. In other words, $\mu$ is \emph{separating}: if two states give the same probabilities for all effects, they must be the same state; and if two effects give the same probabilities for all states, they must be the same effect.
	\item There are trivial effects $I$ and $\emptyset$, the former of which gives probability 1 on all valid states, i.e.\  $\forall S \in \Sigma:\mu(I, S)=1$, and the latter of which gives probability 0 on all valid states, i.e\ $\forall S \in \Sigma:\mu(\emptyset, S)=0$. Moreover, there is for every effect an opposite effect $\neg E$ such that $\forall S \in \Sigma: \mu(\neg E, S)=1-\mu(E,S).$
	\item $\exists S, \forall E, S_1, S_2 : \mu(E, S)=p\mu(E, S_1)+(1-p)\mu(E, S_2)$ for $p\in [0,1]$. Similarly,  $\exists E, \forall S, E_1, E_2 : \mu(E, S)=p\mu(E_1, S)+(1-p)\mu(E_2, S)$. In other words, there is always a state which acts as a convex combination of any two states, and an effect which acts as a convex combination of any two effects, with respect to $\mu$.
\end{enumerate}
Then there exists a vector space $V$ such that
\begin{enumerate}[label=(\roman*)]
	\item The state space and effect spaces are convex subsets of the vector space and its dual: $\Sigma \subseteq V, \mathcal{E} \subseteq V^*$. The trivial effects are contained in the effect space: $(\emptyset|, (I| \in \mathcal{E}$. Moreover, opposite effects are always contained in the effect space, and they are given by  $\forall E \in \mathcal{E}: (\neg E|=(I|-(E| \in \mathcal{E}$.
	\item The state space $\Sigma$ forms the base of a cone $C$ of unnormalized states. The base is picked out by the \emph{normalization functional} $(I|S)=1$. The effect space  $\mathcal{E}$ forms the base of the dual cone $C^*$ of effects $(E| \in V^*$ such that $\forall\  |S) \in C: (E|S)\ge 0$.  
	\item $\mu(E, S) = (E|S)$. In other words, the probability function $\mu$ is just given by an effect as a linear functional acting on a state vector.
	\item A measurement consists of a set of effects such that $\sum_i (E_i| = (I|$. A measurement may have a continuum of outcomes, in which case the sum ought to be read as the appropriate integral.
\end{enumerate}	
\end{theorem}

\begin{remark}
The full theorem is somewhat more technical, and indeed applies even in the case where the vector spaces involved are infinite dimensional. We will confine ourselves to the case where the state and effect spaces live in finite dimensional vector spaces.
\end{remark}
Classical probability theory is an example of a GPT: the state space is the space of probability distributions on $n$ outcomes, i.e.\ the probability simplex, and the effect space is the hypercube dual to the simplex, the space of response functions. The normalization functional $(I|=(1,\dots, 1)$, and so measurements correspond to column stochastic matrices. Quantum theory is also a GPT. The state space of quantum theory is the space of $d \times d$ density matrices, i.e.\ positive semidefinite Hermitian matrices  $\rho \ge 0$ with $\tr(\rho)=1$. In other words, the normalization functional is $\tr(\cdot)$. The effect space consists of positive semidefinite Hermitian matrices $E \ge 0$ with no restriction on their trace. Measurements in quantum theory correspond to \emph{positive operator valued measures} (POVMs), which map events to positive semidefinite operators on Hilbert space or, more simply, which are collections of effects such that $\sum_i E_i=I.$  Both the state space and the effect space consist in PSD matrices: quantum theory is what is known as a \emph{self-dual} theory, for which effects can always be rescaled into states. 

Notice that we can consider quantum states and effects as elements of vector spaces over $\mathbb{R}$. For example, for a qubit, the state space corresponds to the surface and interior of the 2-sphere, which we can see by expanding in the orthogonal basis provided by the Pauli matrices
\begin{align}
\rho &= \frac{1}{2}\Big(I + \tr(\sigma_x\rho) \sigma_x + \tr(\sigma_y\rho)  \sigma_y + \tr(\sigma_z\rho)  \sigma_z\Big).
\end{align}
In this representation, the normalization functional is $(I|=(1,0,0,0)$, and we have 
\begin{align}
 (E|=\frac{1}{2}\begin{pmatrix} \tr(E) & \tr(\sigma_x E) & \tr(\sigma_y E) &\tr(\sigma_z E)\end{pmatrix} && |\rho) = \begin{pmatrix} 	1\\ \tr(\sigma_x \rho)\\\tr(\sigma_y \rho)\\ \tr(\sigma_z \rho) \end{pmatrix}
\end{align}
so that
\begin{align}
P(E)= \tr(E\rho) = (E|\rho).
\end{align}
This same ``Bloch sphere'' construction can be generalized to any dimension $d$ by employing, for example, the generalized Gell-Mann matrices, which are similarly Hermitian, traceless, and orthogonal---and we will take recourse to the analogous representation in non-quantum theories.
  More generally, one can consider GPTs with more exotic state and effect spaces, e.g.\ theories with square state spaces, state spaces corresponding to higher dimensional spheres,  state spaces shaped like ice cream cones, and deduce many illuminating theorems relating their geometries to a diverse set of operationally grounded axioms.

\section{Reference devices}

Let $\{(E_i|\}_{i=1}^n$ be a set of effects which not only correspond to a measurement, but also span the space of effects. Such a measurement is called \emph{informationally complete}: assigning probabilities to the outcomes of such a measurement is equivalent to the assignment of a state. As we've noted before, if the effects form a basis, we'll call such a measurement a \emph{minimal} informationally complete (MIC) measurement\footnote{It has long been known that MICs exist in any GPT \cite{barrett2009finetti}.}. If we gather up the effects as the rows of a matrix
\begin{align}
\B{E}=\begin{pmatrix} (E_1| \\ \vdots \\ (E_n| \end{pmatrix},	
\end{align}
then we can write the linear map from the state space into the probability simplex as 
\begin{align}
P(E) = \B{E}|\rho),
\end{align}
where $P(E)$ is understood as a column vector of probabilities. Similarly, let $\{|S_i)\}_{i=1}^n$ be a set of states which span the state space. Gathering up the states as the columns of a matrix
\begin{align}
\B{S} = \begin{pmatrix} |S_1) & \dots & |S_n) \end{pmatrix},	
\end{align}
we can write the linear map from the effect space into the space of response functions as
\begin{align}
P(S) = (\eta| \B{S},	
\end{align}
where $P(S)$ is understood as a row vector.

An \emph{informationally complete reference device} is a measure-and-prepare device whose effects and states are spanning. First, the device performs the measurement $\{(E_i|\}_{i=1}^n$: upon obtaining the $i$'th outcome, the device prepares the $i$'th state in $\{|S_i)\}_{i=1}^n$. Such a device furnishes a (nearly) probabilistic representation of a GPT. To see this, introduce a \emph{Born matrix} $\Phi$ satisfying the ``Born identity'' \cite{born_identity}
\begin{align}
	\B{S}\Phi \B{E} = I.
\end{align}
Now suppose we have some other measurement $\{(A_i|\}_{i=1}^m$, whose effects we've gathered up into the rows of a matrix $\B{A}$. Using the Born identity, we can rewrite the probabilities $P(A)$ as
\begin{align}
P(A) &= \B{A}|\rho)	\\
&= \B{A}\big[\B{S}\Phi\B{E}\big]|\rho)	\\
&= P(A|E) \Phi P(E)
\end{align}
where $P(A|E)=P(A|S)$ is the conditional probability matrix for the outcomes of $\{(A_i|\}_{i=1}^m$ given the reference preparations $\{|S_i)\}_{i=1}^n$, where we've used the outcomes of $\{(E_i|\}_{i=1}^n$ as a proxy for the state preparation. As  promised, we can use the Born matrix $\Phi$ to transport our probability judgements $P(E)$ with respect to the reference device  to any other measurement. Indeed, instead of working with states and effects at all, we can characterize both preparations and measurements with respect to some chosen reference device. Moreover, since $P(A|E)P(E)$ is just the usual law of total probability, we can think of the Born matrix $\Phi$ as providing a deformation of the LTP. 

How can one find such a matrix $\Phi$? Consider the following:
\begin{theorem}
Given an informationally complete set of states and effects, a Born matrix satisfying $\B{S}\Phi\B{E}=I$ is equivalent to a $\{1\}$-inverse of $P(E|E)=\B{ES}$.
\end{theorem}
\begin{proof}

Multiply the Born identity on the left by $\B{E}$ and on the right by $\B{S}$ to obtain
\begin{align}
	\B{E}\big[\B{S}\Phi \B{E}\big]\B{S} &= \B{ES}\\
	P(E|E) \Phi P(E|E) &= P(E|E)
\end{align}
Here $P(E|E)$ is the \emph{self-conditional probability matrix} containing the probabilities for obtaining a reference outcome given that the reference device's own preparations were routed into it. For simplicity, let $P\equiv P(E|E)$. $P\Phi P = P$ is the defining equation of a species of generalized matrix inverse known as a $\{1\}$-inverse\footnote{For concreteness, all $\{1\}$-inverses can be obtained from the singular value decomposition: the Moore-Penrose pseudoinverse is a special case.} \cite{benisrael_greville}. Denoting the dimensionality of the GPT as $r$, informational completeness implies that $\B{E} \in \mathbb{R}^{n \times r}$ will have full column rank, and thus have a left inverse $\B{E}^L$; similarly, $\B{S} \in \mathbb{R}^{r \times n}$ will have full row rank, and thus have a right inverse $\B{S}^R$. Acting on $P\Phi P=P$  with $\B{E}^L$ on the left and $\B{S}^R$ on the right returns us to $\B{S}\Phi \B{E}=I$. Thus a $\Phi$ satisfying $P\Phi P = P$ will also satisfy the Born identity $\B{S}\Phi \B{E}= I_r$, and vice versa.
\end{proof}

\begin{remark}
In the case of a MIC, $P$ is invertible and so $\Phi= P^{-1}$.	
\end{remark}

Finally, letting $\BT{S}=\B{S}\Phi$ and $\BT{E}=\Phi \B{E}$, we can also write the Born identity as
\begin{align}
\BT{S}\B{E} = I=\B{S}\BT{E},
\end{align}
where we'll call the columns of $\BT{S}$  \emph{dual states} and the rows of $\BT{E}$  \emph{dual effects}. In the language of frame theory \cite{waldron}, we can think of $\B{E}$ as an analysis operator, analyzing states into probabilities, and $\BT{S}$ as a choice of synthesis operator, taking probabilities back to states. Similarly, we can view $\BT{E}$ is an analysis operator into \emph{quasi-probabilities}, which may be negative, even as they still sum to 1; and $\B{S}$ as a choice of synthesis operator, taking quasi-probabilities back to states \cite{born_identity}. In a nod to Wigner, we'll denote quasi-probability vectors as $W(E)=\BT{E}|\rho)$.

\section{The Urgleichung}

In what follows, it will prove useful to bring our states and effects into a standard \emph{Bloch form} such that the normalization functional is $(I|_B= (1,0,\dots, 0)$. This can always be accomplished, as the following lemma shows.
\begin{lemma} 
Any GPT can be brought, not uniquely, into Bloch form. 	
\end{lemma}
\begin{proof}
 To see this, consider two fundamental projectors
\begin{align}
\Pi_1 = \frac{|I)(I|}{(I|I)} && \Pi_0 = I - \Pi_1.
\end{align}
The former projects onto the ``traceful'' subspace, which is 1-dimensional, and the latter projects onto the ``traceless'' subspace, which is $(r-1)$-dimensional. Here $r$ is the dimension of the GPT. Let $\{b_i\}_{i=1}^{r-1}$ be a set of $r$-vectors which form an orthonormal basis for the traceless subspace. Each will, of course, be orthogonal to $(I|$. Then the \emph{Bloch transformation}
\begin{align}
\mathcal{B} = \begin{pmatrix} (I| \\ b_1 \\ \vdots \\b_{r-1}\end{pmatrix}	
\end{align}
 will be an invertible $r \times r $ matrix which brings a state (not uniquely) to Bloch form
 \begin{align}
 \mathcal{B}|\rho) =  \left(\begin{array}{c}  
1  \\
\hline
\rho_0
\end{array}\right),
 \end{align}
which we write as a block vector. Meanwhile, effects transform as $(\eta|\mathcal{B}^{-1}$ so that all probabilities are preserved, along with convex combinations of states and effects.
\end{proof} 
 
Thus in what follows, we will always assume that $(I|=(1,0, \dots, 0)$, and without loss of generality, we can write an arbitrary measurement as a block (Bloch) matrix
\begin{align}
\B{E}=
\left(\begin{array}{c|ccc}  
& & & \\
w & & \B{Z}^\dagger \\
& & &\\
\end{array}\right)
\end{align}
such that $\sum_{i=1}^n (E_i| = (I|$, or more simply,  $u^\dagger \B{E} = (I|$ where $u^\dagger = (1,\dots, 1)$. Here $w$ is a column vector of weights and $\B{Z}^\dagger \in \mathbb{R}^{n \times (r-1)}$. Notice this implies $\sum_{i=1}^n w_i = 1$, while the columns of $\B{Z}^\dagger$ sum to 0. The weights $w_i$ can be interpreted as \emph{biases}: we have $\B{E}|I) = w$, so that the biases are just the outcome probabilities given the state of maximal ignorance.

With this in hand, the map from the state space into the probability simplex can be written as an \emph{affine transformation} from the traceless subspace instead.
\begin{align}
P(E) &= \B{E}|\rho) \\
&= 	\left(\begin{array}{c|ccc}  
& & & \\
w & & \B{Z}^\dagger \\
& & &\\
\end{array}\right) \left(\begin{array}{c}  
1  \\
\hline
\rho_0
\end{array}\right)\\
&= w + \B{Z}^\dagger \rho_0.
\end{align}
Recall the Born identity $\B{S}\BT{E}=I$, where $\BT{E}$ is the map from states into quasi-probabilities $W(E)=\BT{E}|\rho)$. If we were to demand that the quasi-probability representation be as close as possible to the probability representation, we could do no better than to assume that
\begin{align}
W(E) = w+ \alpha \B{Z}^\dagger \rho_0	
\end{align}
for some choice of a single parameter $\alpha$. Then
\begin{align}
W(E) &= \alpha(P(E)-w) + w\\
&= \alpha P(E) + (1-\alpha)w,
\end{align}
which tells us that a simple mixture of the probabilities and the biases yield quasi-probabilities. Since $\BT{E}=\Phi\B{E}$, we obtain the Protourgleichung
\begin{align}
P(A_i) &= \sum_{jk}P(A_i|E_j)\Phi_{jk} P(E_k)\\
&=	\sum_j P(A_i|E_j)\Big[\alpha P(E_j) + (1-\alpha)w_j\Big].
\end{align}

 We now establish the following theorem:
\begin{theorem}
Given an arbitrary\footnote{For simplicity of exposition, we have confined ourselves to measurements with a finite number of outcomes, although the proof generalizes to measurements with a continuum of outcomes.} informationally complete measurement in a finite dimensional GPT, one can always choose a set of reference states which allows a quasi-probability representation of the form $W(E)=\alpha P(E) + (1-\alpha)w$.
\end{theorem}
\begin{proof}
Let 
\begin{align}
\B{S} = \left(\begin{array}{ccc}  
1 & \dots & 1 \\
\hline
& \\
& \B{X} \\
& &\\
\end{array}\right)	&& 
\BT{E}=
\left(\begin{array}{c|ccc}  
& & & \\
w & & \alpha\B{Z}^\dagger \\
& & &\\
\end{array}\right).
\end{align}
The Born identity requires that $\B{S}\BT{E}=I$. We have
\begin{align}
 \left(\begin{array}{ccc}  
1 & \dots & 1 \\
\hline
& \\
& \B{X} \\
& &\\
\end{array}\right)\left(\begin{array}{c|ccc}  
& & & \\
w & & \alpha\B{Z}^\dagger \\
& & &\\
\end{array}\right) &= \begin{pmatrix}u^\dagger w & \alpha u^\dagger \B{Z}^\dagger \\
\B{X}w & \alpha \B{X}\B{Z}^\dagger
 \end{pmatrix}\\
 &= \begin{pmatrix}1& 0 \\
\B{X}w & \alpha \B{X}\B{Z}^\dagger
 \end{pmatrix}
\end{align}
where the last follows from $u^\dagger w =1$ and $u^\dagger \B{Z}^\dagger = 0$. To satisfy the Born identity, then, our states $\B{S}$ must satisfy
\begin{align}
\B{X}w=0 && \alpha\B{X}\B{Z}^\dagger = I_{r-1}.	
\end{align}
On the one hand, $\B{Z}^\dagger \in \mathbb{R}^{n\times r-1}_{r-1}$: informational completeness implies that $\B{Z}^\dagger$ has full column rank, and therefore has (not necessarily uniquely) a left inverse $(\B{Z}^\dagger)^L$. This will ultimately furnish our $\B{X}$, but there is nothing as yet guaranteeing that the columns of $\B{S}$ will correspond to valid states in the GPT. In other words, we further require that $(\eta|S_i) \ge 0$ for all effects $(\eta|$ and for each reference state $|S_i)$. In order to accommodate this, we could simply rescale $(\B{Z}^\dagger)^L$ so that $\bold{X}=\frac{1}{\alpha}(\B{Z}^\dagger)^L$ where $\frac{1}{\alpha}$ is chosen so that
\begin{align}
\forall \eta, i: (\eta|S_i) = 	\left(\begin{array}{c|c}  
\eta_1 & \eta_0^\dagger
\end{array}\right)\left(\begin{array}{c}  
1 \\
\hline\\[-10pt]
\frac{1}{\alpha}z_i^L \\
\end{array}\right) &= \eta_1 + \frac{1}{\alpha}\eta_0^\dagger z_i^L \ge 0,
\end{align}
where $z_i^L$ denotes the $i$'th column of $(\B{Z}^\dagger)^L$. Since the weights $\eta_1$ on any effect are non-negative, we can always make $(\eta|S_i)\ge0$ by choosing $\alpha \ge 1$ sufficiently large since any negativity enters in through  $\eta_0^\dagger z_i^L$. That said, there is one subtlety here: it may be that we can take $\alpha \leq -1$ and obtain an alternative set of valid states. Thus we merely require that $|\alpha| \ge 1$.

But we must also have that $\B{X}w=0$: since the rows of $\alpha\B{X}$ form a biorthogonal system with the columns of $\B{Z}^\dagger$, we require that $w$ be linearly independent of the columns of $\B{Z}^\dagger$ in order that the rows of $\B{X}$ can be chosen orthogonal to $w$. But this is clearly the case since the columns of $\B{Z}^\dagger$ all sum to 0, so that no linear combination of them can yield $w$, which sums to 1. 

Putting it all together, we have
\begin{align}
\B{E}=
\left(\begin{array}{c|ccc}  
& & & \\
w & & \B{Z}^\dagger \\
& & &\\
\end{array}\right) && \B{S} = \left(\begin{array}{ccc}  
1 & \dots & 1 \\
\hline
& \\
& \B{X} \\
& &\\
\end{array}\right)
\end{align}
such that
\begin{align}
u^\dagger \B{E}=(I| && \B{X} w = 0 && \alpha \B{X} \B{Z}^\dagger = I_{r-1} 
\end{align}
where the reference effects have been chosen to be informationally complete, and the states, which derive from different choices of left inverses of $\B{Z}^\dagger$, inherit the informational completeness of the effects. We have therefore shown that a set of reference states can always be found regardless of the reference measurement which yields a quasi-probability representation of the desired form since $\BT{E}|\rho)=w+ \alpha \B{Z}^\dagger \rho_0=	\alpha P(E) + (1-\alpha)w$.
\end{proof}

\begin{remark}
The condition $\B{X}w=0$ amounts to $\B{S}w = |I)$: the weighted average of the reference states is the maximally mixed state\footnote{In a Bloch representation, the maximally mixed state $|M)=(I|^\dagger$. In other representations, it may be necessary to renormalize so that $|M)=\frac{|I)}{(I|I)}$.}. In a self-dual theory, and for an unbiased measurement, this implies that the reference states themselves correspond to a measurement. 
\end{remark}

Another perspective on this construction is provided by considering the channel operator corresponding to reference device $\mathcal{C}= \sum_{i=1}^n |S_i)(E_i|$.
\begin{corollary}
	A reference device which furnishes a quasi-probability representation $W(E)=\alpha P(E) + (1-\alpha)w$ corresponds to a depolarizing channel with parameter $\frac{1}{\alpha}$. The converse also holds.
\end{corollary}
\begin{proof}
Working out	the form of the channel operator, we find
\begin{align}
 \B{S}\B{E} &=\left(\begin{array}{ccc}  
1 & \dots & 1 \\
\hline
& \\
& \B{X}\\
& &\\
\end{array}\right)\left(\begin{array}{c|ccc}  
& & & \\
w & & \B{Z}^\dagger \\
& & &\\
\end{array}\right)\\
&= \begin{pmatrix}u^\dagger w & u^\dagger \B{Z}^\dagger \\ \B{X} w & \B{X} \B{Z}^\dagger \end{pmatrix}= \begin{pmatrix}1 & 0 \\ 0 & \frac{1}{\alpha}I_{r-1}\end{pmatrix}\\
&=\frac{1}{\alpha}I+\left(1-\frac{1}{\alpha}\right) \frac{|I)(I|}{(I|I)}.
\end{align}
This implies that the reference devices acts as a \emph{depolarizing channel} with parameter $\frac{1}{\alpha}$:
\begin{align}
|\rho) \mapsto \frac{1}{\alpha}|\rho)+\left(1-\frac{1}{\alpha}\right)|M).
\end{align}
It is not hard to see that the argument can be read in reverse: if one chooses a set of states such that the reference device forms a depolarizing channel---and one can always make such a choice---then one can employ such a quasi-probability representation. 
\end{proof}

\begin{remark}
It may be useful to see step by step ``why'' the Protourgleichung works. We begin by inserting the Born identity into the expression for $P(A)$,
\begin{align}
P(A) &= \B{A}|\rho)	= \B{A}\left[\B{S}\Phi \B{E}\right] |\rho).
\end{align}
If $\Phi$ acts as an Protourgleichung, we have
\begin{align}
&=\B{A}\B{S}\Big[\alpha \B{E}|\rho)+(1-\alpha)w\Big]=\B{A}\Big[\alpha\B{SE}|\rho) +(1-\alpha)\B{S}w\Big].
\end{align}
But we know that $\B{SE}$ is just the channel operator, which is depolarizing with parameter $\frac{1}{\alpha}$, and $\B{S}w=|M)$, the maximally mixed state. The effect, then, is simply to balance depolarization with repolarization:
\begin{align}
&=\B{A}\Bigg[\alpha\left[ \frac{1}{\alpha}|\rho)+\left(1-\frac{1}{\alpha}\right)|M)\right] +(1-\alpha)|M)\Bigg]\\
&=\B{A}\Big[|\rho)- (1-\alpha)|M) +(1-\alpha)|M)\Big]\\
&= \B{A}|\rho)
\end{align}
	
\end{remark}

\begin{remark}
 The case of $\alpha=1$ is interesting: then $W(E)=P(E)$, so that the Protourgleichung collapses back to the law of total probability. Since $\BT{E}=\B{E}$, the channel operator is simply the identity $\mathcal{C}=\B{SE}=I$, which implies that $P^2=\B{E}\big[ \B{SE}\big] \B{S}=\B{ES}=P$: the self-conditional probability matrix is a projector. Moreover, notice that
\begin{align}
	\forall \eta, \rho: (\eta|\rho) = \Big[(\eta|\B{S}\Big]\Big[\B{E}|\rho)\Big]=P(S)  P(E).
\end{align}
 Thus $\B{S}$ provides a linear map which embeds the effect space into the hypercube dual to the probability simplex, and $\B{E}$ provides a linear map which embeds the state space into the probability simplex itself in such a way that all probabilities $P(\eta|\rho)$ are preserved. We've therefore constructed a \emph{simplex embedding}, or equivalently a non-contextual ontological model \cite{Schmid_2021, selby2022opensource}: in short, if we can take $\alpha=1$, then our GPT must be equivalent to a classical theory, i.e.\ with an epistemic restriction.
\end{remark}

For a depolarizing reference device constructed according to the above prescription, the self-conditional probability matrix takes the general form
\begin{align}
P(E|E) &= \B{ES} = 		\left(\begin{array}{c|ccc}  
& & & \\
w & & \B{Z}^\dagger \\
& & &\\
\end{array}\right)\left(\begin{array}{ccc}  
1 & \dots & 1 \\
\hline
& \\
& \B{X} \\
& &\\
\end{array}\right)\\
&= wu^\dagger + \B{Z}^\dagger \B{X}.
\end{align}
We then have the following theorem:
\begin{theorem}
\label{equiangular}
The self-conditional probability matrix $P(E|E)$ of a depolarizing reference device whose measurement is an unbiased MIC is equiangular.	
\end{theorem}
\begin{proof} Supposing the dimensionality of the GPT is $r$, then $\B{E}$, $\B{S}$, and $P(E|E)$ are all $r\times r$ matrices of full rank. Since by construction, $\B{X}\B{Z}^\dagger = \frac{1}{a}I_{r-1}$, and since $\B{X}\B{Z}^\dagger$ and $\B{Z}^\dagger\B{X}$ have the same non-zero spectrum, we know that $\B{Z}^\dagger\B{X}$ has $r-1$ eigenvalues equal to $\frac{1}{\alpha}$ and a single eigenvalue equal to 0. Thus $\alpha \B{Z}^\dagger \B{X}$ is a projector. Moreover, since $u^\dagger\B{Z}^\dagger=0$, we have that
\begin{align}
\tr\left( \big[ wu^\dagger \big] \big[\alpha \B{Z}^\dagger \B{X}\big]\right)=0:	
\end{align}
in other words, $\alpha \B{Z}^\dagger \B{X}$ and $wu^\dagger$, which is itself a projector $(wu^\dagger wu^\dagger = wu^\dagger$, since the weights sum to 1), must project onto orthogonal subspaces. We thus conclude that
\begin{align}
	\alpha \B{Z}^\dagger \B{X} = I - wu^\dagger.
\end{align}
With this in hand, we have
\begin{align}
P(E|E) &= wu^\dagger + \frac{1}{\alpha}(I - wu^\dagger)\\
	&= \frac{1}{\alpha}I + \left(1-\frac{1}{\alpha}\right)wu^\dagger.
\end{align}
Supposing that the initial measurement is unbiased, we have $w_i=\frac{1}{r}$, and thus 
\begin{align}
P(E|E) &= 	\frac{1}{\alpha}I + \left(1-\frac{1}{\alpha}\right)\frac{1}{r}uu^\dagger.
\end{align}
This is an equiangular conditional probability matrix, with one unique value repeated along the diagonal, and another unique value on the off-diagonal. 
\end{proof}

\section{Morphophoricity}

The keystone of the construction in the previous section was the demand that
\begin{align}
\alpha \B{X}\B{Z}^\dagger = I_{r-1}	
\end{align}
so that $\B{X} = \frac{1}{\alpha}(\B{Z}^\dagger)^{L}$ where $\alpha$ is chosen so that the resulting states are non-negative on all effects.  Suppose, however, that we demand that the reference states take the particular form
\begin{align}
\B{S} = 	\left(\begin{array}{ccc}  
1 & \dots & 1 \\
\hline
& \\
& c\B{Z}W^{-1} \\
& &\\
\end{array}\right),
\end{align}
where $W^{-1}$ is the diagonal matrix with the reciprocals of the weights $\frac{1}{w_i}$ along the diagonal. If the constant $c$  were equal to 1, then this would amount to taking $|S_i) = \frac{1}{w_i}(E_i|^\dagger$, that is, to taking the reference states proportional to effects. This is always possible in a self-dual GPT: otherwise, employing the same trick as before, one may choose the constant $c$ in order to depolarize the prospective states until they are non-negative on all effects. Either way, we'll call such a reference device a $\emph{parallel update}$ reference device. If such a device is to be depolarizing, we must have
\begin{align}
\alpha c\B{Z}W^{-1}\B{Z}^\dagger = I_{r-1}	,
\end{align}
or $\B{X}=c\B{Z}W^{-1} = \frac{1}{\alpha}(\B{Z}^\dagger)^L$: since $c$ is fixed by the demand that $c\B{Z}W^{-1}$ furnish valid states, this fixes the value of $\alpha$. In other words, we require that the columns of $\B{Z}^\dagger$ (which notice, cut against the grain of the effects) must be orthonormal with respect to the metric provided by $\mathcal{M}=\alpha cW^{-1}$. 

If this is possible, we have the following interesting consequence. Recalling that $P(E)=\B{E}|\rho)=w+\B{Z}^\dagger \rho_0$, we have
\begin{align}
\big{\lVert}P(E)_\rho - P(E)_\sigma \big{\lVert}^2_{\mathcal{M}} &= \big{\lVert}w+ \B{Z}^\dagger \rho_0 - w-\B{Z}^\dagger \sigma_0 \big{\lVert}^2_{\mathcal{M}}\\
&=\big{\lVert}\B{Z}^\dagger (\rho_0 - \sigma_0)\big{\lVert}^2_\mathcal{M}\\
&= (\rho_0 - \sigma_0)^\dagger\B{Z} \mathcal{M} \B{Z}^\dagger(\rho_0 - \sigma_0)\\
&= (\rho_0 - \sigma_0)^\dagger\alpha c\B{Z}W^{-1}\B{Z}^\dagger(\rho_0 - \sigma_0)\\
&= \big{\lVert}\rho_0-\sigma_0 \big{\lVert}^2.
\end{align}
In other words, up to the diagonal metric $\mathcal{M}$, the measurement map from the traceless subspace into the probability simplex preserves distances. We'll call such a measurement \emph{weight-morphophoric}. For an unbiased measurement, we have $\B{Z}\B{Z}^\dagger = \frac{1}{\alpha cn}I$ so that
\begin{align}
	\big{\lVert}P(E)_\rho - P(E)_\sigma \big{\lVert}^2&= \frac{1}{\alpha cn} \big{\lVert}\rho_0-\sigma_0 \big{\lVert}^2:
\end{align}
the measurement map is then a \emph{similarity} with respect to the Euclidean metric on both the probability simplex and the traceless subspace. Such a measurement is properly \emph{morphophoric}.

\begin{remark}
``Morphophoric'' in Greek means \emph{form-bearing}: indeed, a morphophoric measurement embeds a state space into the probability simplex in a shape preserving way.	The terminology was introduced in \cite{morphophoric}. A rigorous mathematical treatment of the subject can be found in \cite{morpho}. The essence of morphophoricity lies in the demand that $\B{Z}\B{Z}^\dagger \propto I$: the measurement furnishes a so-called \emph{tight frame} for the traceless subspace. Morphophoric measurements in general yield Urgleichung-like expressions which are somewhat more sophisticated\footnote{Indeed, which are not necessarily Protourgleichung in our sense, as the post-measurement states may depend on the input state.}; we confine our interest here to the very simplest deformations of the law of total probability which correspond to quasi-probability representations of the form $W(E)=\alpha P(E) + (1-\alpha)w$. This motivates our introduction of the notion of weight-morphophoricity, that is, the demand that $\B{Z}W^{-1}\B{Z}^\dagger \propto I$ in order to handle biased measurements on equal footing with unbiased.
\end{remark}

\begin{lemma}
Weight-morphophoric measurements exist in any GPT for any number of outcomes $n$ and arbitrary choice of biases.
\end{lemma}
\begin{proof}
Suppose the dimensionality of the GPT is $r$. Pick any set of $r-1$ orthogonal $n$-dimensional vectors each of which is orthogonal to $u = (1,\dots,1)^\dagger$. This can always be done. Arrange them into the columns of a matrix $\B{R^\dagger}$ so that $\B{R}\B{R}^\dagger = I_{r-1}$. Pick an arbitrary set of weights such that $\sum_i w_i = 1$. We can then easily construct $(\B{R}^\dagger)^\prime$ so that 
\begin{align}
	\B{R}\B{R}^\dagger = \left[\B{R}W^{\frac{1}{2}}\right]W^{-1} \left[W^{\frac{1}{2}}\B{R}^\dagger\right]=\B{R}^\prime W^{-1}(\B{R}^\dagger)^\prime= I_{r-1}.
\end{align}
There is no guarantee that $(\B{R}^\dagger)^\prime$ will furnish valid effects and $\B{R}^\prime W^{-1}$ will furnish valid states. So depolarize them by picking parameters $\alpha$ and $c$ so that $\sqrt{\frac{c}{\alpha}}\B{R}^\prime W^{-1}$ delivers states non-negative on all effects and $\frac{1}{\sqrt{\alpha c}}(\B{R}^\dagger)^\prime$ delivers effects non-negative on all states\footnote{And indeed, so that $\alpha c>0$.}. Then
\begin{align}
	\alpha c\left[\frac{1}{\sqrt{\alpha c}}\B{R}^\prime\right] W^{-1}\left[ \frac{1}{\sqrt{\alpha c}}(\B{R}^\dagger)^\prime\right]= \alpha c \B{Z}W^{-1}\B{Z}^\dagger=  I_{r-1}.
\end{align}
Thus we've constructed a weight-morphophoric measurement with
\begin{align}
\B{E}=	\left(\begin{array}{c|ccc}  
& & & \\
w & & \B{Z}^\dagger \\
& & &\\
\end{array}\right) && \B{S}=\left(\begin{array}{ccc}  
1 & \dots & 1 \\
\hline
& \\
& c\B{Z}W^{-1} \\
& &\\
\end{array}\right).
\end{align}
The fact that the weights sum to 1 and the original $n$-dimensional vectors were orthogonal to $u$ guarantees that $\sum_i (E_i| = (I|$.  Informational completeness is inherited from their original vectors' orthonormality.
\end{proof}

Suppose we have a weight-morphophoric measurement, i.e.\ such that the states proportional to the reference effects  (modulo $c$) furnish a depolarizing channel. Regardless, we are not limited to this particular choice of reference states: any $\B{S}$ built from a left inverse of $\B{Z}^\dagger$ will suffice. There is a special case, however.

\begin{theorem}
For a weight-morphophoric MIC reference measurement, the choice of reference states which furnish a depolarizing channel is fixed up to sign.
\end{theorem}
\begin{proof}
Suppose there were some $\B{X}$ such that both
\begin{align}
\alpha c \B{Z}W^{-1}\B{Z}^\dagger = I_{r-1} && \beta \B{X}\B{Z}^\dagger = I_{r-1}	
\end{align}
were true: here  $\B{X}, \B{Z} \in \mathbb{R}^{(r-1)\times r}_{r-1}$. But this means
\begin{align}
(\alpha c \B{Z}W^{-1}	-\beta\B{X})\B{Z}^\dagger = 0,
\end{align}
which implies that the $r-1$ rows of $(\alpha c \B{Z}W^{-1}	-\beta\B{X})$ are each orthogonal to the $r-1$ columns of $\B{Z}^\dagger$. A vector in an $r$-dimensional space can only be orthogonal to at most $r-1$ vectors. The columns of $\B{Z}^\dagger$ are linearly independent, and so each of the $r-1$ rows of $(\alpha c \B{Z}W^{-1}	-\beta\B{X})$ is orthogonal to an $r-1$ dimensional subspace. So either they all live in the same 1-dimensional subspace, or else $\alpha c \B{Z}W^{-1}	=\beta\B{X}$. Thus we conclude that for a weight-morphophoric MIC, the choice of depolarizing reference states is essentially fixed (up to $c$): it may, however, be the case that we can take $\B{X} =-c\B{Z}W^{-1}$ and $\beta=-\alpha$, hence the reference states are fixed up to sign.
\end{proof}

\begin{remark}
 Consider the following SIC-POVM in $d=2$. The POVM elements can be written in terms of Pauli matrices
\begin{align}
E_0 = \frac{1}{4}\left(I + \frac{1}{\sqrt{3}}(\sigma_x + \sigma_y + \sigma_z\right) && E_2 &= \frac{1}{4}\left(I + \frac{1}{\sqrt{3}}(-\sigma_x - \sigma_y + \sigma_z\right)	\\
E_1 = \frac{1}{4}\left(I + \frac{1}{\sqrt{3}}(\sigma_x - \sigma_y - \sigma_z\right) && E_3&= \frac{1}{4}\left(I + \frac{1}{\sqrt{3}}(-\sigma_x + \sigma_y - \sigma_z\right),
\end{align}
or more compactly in the Bloch representation
\begin{align}
\B{E} = \left(\begin{array}{c|ccc} 
 \frac{1}{4}u & \B{Z}^\dagger
 \end{array}\right)	&& \B{Z}^\dagger = \frac{1}{4\sqrt{3}}\begin{pmatrix} 1 & 1 & 1 \\
 1 & -1 & -1\\
-1 & -1 & 1 \\
-1 & 1 & -1  \end{pmatrix}.
\end{align}
We could choose states proportional to effects $
\B{S} = \begin{pmatrix} u^\dagger \\ \hline 4\B{Z}\end{pmatrix}$, in which case $\B{Z}W^{-1}\B{Z}^\dagger = 4\B{Z}\B{Z}^\dagger= \frac{1}{3}I$ so that $\alpha=3$ (and $c=1$). Then $\B{SE}$ corresponds to a depolarizing channel with parameter $\frac{1}{3}$: 
\begin{align}
|\rho) \mapsto \frac{1}{3}|\rho) + \frac{2}{3}|M)
\end{align}
so that $P(A) = P(A|E)\big[3P(E)-\frac{1}{2}u\big]$. Alternatively, however, we could choose to prepare the states \emph{antipodal} to the aforementioned, which also form a SIC, that is, we could take $
\B{S}^\prime = \begin{pmatrix} u^\dagger \\ \hline -4\B{Z}\end{pmatrix}$, in which case  $\B{X}=-\B{Z}W^{-1}=-4\B{Z}$ and $\beta=-\frac{1}{3}$. Then 
\begin{align}
|\rho) \mapsto -\frac{1}{3}|\rho) + \frac{4}{3}|M)
\end{align}
with $P(A)=P(A|E) \big[-3P(E)+u\big]$. In quantum theory,  complete positivity for a depolarizing channel demands that $-\frac{1}{d^2-1} \leq \frac{1}{\alpha}\leq 1$. Indeed, when $d=2$, the lower bound becomes $-\frac{1}{3}$, which this latter channel saturates. Thus for both choices of post-measurement states, we have a valid depolarizing channel: however it is worth remarking that the latter choice also leads to an inversion of the state space. For an arbitrary GPT, however, it may or may not be the case that the channel with an inverted depolarization parameter is a valid channel.
\end{remark}

\begin{corollary}
Suppose we have an (unbiased) morphophoric MIC reference measurement: its depolarizing reference states are fixed up to sign to be proportional to effects (modulo $c$). By Theorem \ref{equiangular}, its self-conditional probability matrix must be equiangular. Thus all morphophoric MIC measurements are SICs, that is, \emph{symmetric informationally complete} measurements which form a regular simplex in the effect space. Note, however, that such measurements need not be extremal, that is, lying on the boundary of the effect space. It is easy to see that non-extremal SICs exist in any GPT: one may simply embed a sufficiently small $(r-1)$-simplex (with $r$ vertices) in the effect space such that the effects are non-negative on all states.
\end{corollary}

We can now relate this construction to quantum theory.

\begin{theorem}
In quantum theory over $\mathbb{C}$, pure weight-morphophoric reference devices correspond to weighted complex projective 2-designs.
\end{theorem}
\begin{proof}
For a depolarizing reference device, we have
\begin{align}
\mathcal{C} = \B{SE}=\sum_{i=1}^n |S_i)(E_i|=	\frac{1}{\alpha}I+\left(1-\frac{1}{\alpha}\right) \frac{|I)(I|}{(I|I)}.
\end{align}
Given that we are working with quantum theory over $\mathbb{C}$, we can interpret this expression  in terms of vectorized operators, which will be elements of $\mathbb{C}^{d^2}$ \cite{caves_etal}. Let $|A) = (A \otimes I )\sum_{i=1}^d |i,i\rangle$ be the (row) vectorization of $A$. It consists of the rows of $A$ laid end to end to form a column vector. Then
\begin{align}
(A|B) &= \Bigg[ \sum_{i=1}^d \langle i,i| (A^{\dagger} \otimes I) \Bigg]		 \Bigg[ (B\otimes I) \sum_{j=1}^d |j,j\rangle \Bigg] \\
&= \sum_{i,j} \langle i|A^{\dagger} B|j \rangle \otimes \langle i|j \rangle= \sum_i \langle i | A^{\dagger} B |i \rangle= \tr(A^{\dagger}B).
\end{align}
 In other words, $(A|B)$ reproduces the Hilbert-Schmidt inner product $\tr(A^\dagger B)$ on operators. Note that in this case $(E|$ really does mean the conjugate transpose of $|E)$. 
 
Noting that $(I|$ is the vectorized identity operator while $\frac{1}{d}|I)$ corresponds to the maximally mixed state, so that $(I|I)=d$, we have in particular
\begin{align}
\label{quantumC}
\sum_{i=1}^n |S_i)(E_i|= \frac{1}{\alpha}I_{d^2}	+ \left(1-\frac{1}{\alpha}\right)\frac{1}{d}|I_d)(I_d|.
\end{align}
Taking the trace of the right-hand side of Eq. \!(\ref{quantumC}), we find
\begin{align}
\sum_{i=1}^n \tr(E_i S_i) &= \frac{1}{\alpha}(d^2) +\left(1-\frac{1}{\alpha}\right)\frac{1}{d}(d)=\frac{1}{\alpha}(d^2-1)+1.
\end{align}
Quantum theory over $\mathbb{C}$ is a self-dual theory, and so we can always take reference states directly proportional to effects: thus let $E_i = \tr(E_i)S_i$---this along with the depolarizing assumption amounts to weight-morphophoricity. Moreover, let us assume that our states are pure so that $\tr(S_i^2)=1$. Then taking the trace of the left hand side of Eq. \!(\ref{quantumC}), we find 
\begin{align}
	\sum_{i=1}^n \tr(E_i S_i) &= \sum_{i=1}^n \tr(E_i)\tr(S_iS_i) = \sum_{i=1}^n \tr(E_i)=\tr\left(\sum_{i=1}^n E_i\right)=\tr(I)=d.
\end{align}
Equating these two expressions fixes the value of $\alpha$:
\begin{align}
	\frac{1}{\alpha}(d^2-1)+1 &= d \Longrightarrow
	\alpha = d+1.
\end{align}
Thus our channel operator becomes
\begin{align}
\mathcal{C}=\sum_{i=1}^n \tr(E_i)|S_i)(S_i| &= \frac{1}{d+1}I_{d^2}	+ \left(1-\frac{1}{d+1}\right)\frac{1}{d}|I_d)(I_d|\\
&= \frac{1}{d+1}\Big(I_{d^2}+|I_d)(I_d|\Big).
\end{align}
Normalizing so that the trace of both sides is 1, we find
\begin{align}
\label{todesign}
	\sum_{i=1}^n w_i|S_i)(S_i| = \frac{1}{d(d+1)}\Big(I_{d^2} + |I_d)(I_d|\Big),
\end{align}
where $w_i = \frac{\tr(E_i)}{d}$ so that $\sum_i w_i = 1$.  
 Considering the right hand side of Eq.\! (\ref{todesign}), notice that we can rewrite it using the definition of vectorization as
 \begin{align}
	I_{d^2} + |I_d)(I_d| &= I_d \otimes I_d + \sum_{jk}|j,j\rangle\langle k, k|.
\end{align}
Taking the partial transpose of the second tensor factor and multiplying by $\frac{1}{2}$, we find:
\begin{align}
	 \frac{1}{2}\Big(I_d \otimes I_d + \sum_{jk}|j,k\rangle\langle k, j| \Big)=  \frac{1}{2}\Big(I_d \otimes I_d + \text{SWAP}\Big)= \Pi_{\text{sym}^2},
\end{align}
where $\Pi_{\text{sym}^2}$ is the projector onto the permutation symmetric subspace of two tensor factors. Considering the left hand side of Eq. \!(\ref{todesign}), we must also take the transpose of the second tensor factor. It is a useful lemma that for a pure states $S = |\psi\rangle\langle \psi|$, we have $|S)(S| = S \otimes S^T$. Recalling the factor of $\frac{1}{2}$, we obtain
\begin{align}
\sum_{i=1}^n w_i S_i \otimes S_i = \frac{2}{d(d+1)}\Pi_{\text{sym}^2}.	
\end{align}
Meanwhile, Schur's lemma tells us that the integral over tensor powers of pure states is
\begin{align}
\int |\psi\rangle\langle\psi|^{\otimes t} d\psi = \frac{\Pi_{\text{sym}^t}}{\tr(\Pi_{\text{sym}^t})} = \frac{1}{\binom{t+d-1}{d-1}}\Pi_{\text{sym}^t},
\end{align}
 which in the case of $t=2$, delivers us
 \begin{align}
 	\int |\psi\rangle\langle\psi|^{\otimes 2} d\psi = \frac{2}{d(d+1)}\Pi_{\text{sym}^2}=\sum_{i=1}^n w_i S_i \otimes S_i.
 \end{align}
Indeed, this is the defining equation for a weighted complex projective 2-design, showing that such designs coincide precisely with pure weight-morphophoric measure-and-prepare reference devices \cite{scott}.
\end{proof}

\begin{corollary}
	For a pure weight-morphophoric reference device in quantum theory, we have $\alpha=d+1$ so that $W(E)=\alpha P(E) + (1-\alpha)w$ becomes
	\begin{align}
	W(E) = (d+1)P(E) -dw.
	\end{align}
 For an unbiased measurement:
 	\begin{align}
	W(E) = (d+1)P(E) -\frac{d}{n}u.
	\end{align}
For a MIC, with $n=d^2$:
 	\begin{align}
	W(E) = (d+1)P(E) -\frac{1}{d}u.
	\end{align}
It has long been known that the unbiased complex projective 2-designs with the minimal number of elements correspond to pure SICs with $P(E_i|S_j)= \frac{d\delta_{ij} +1}{d(d+1)}$, if indeed SICs always exist. We thus recover the Urgleichung:
\begin{align}
P(A_i) = \sum_{j=1}^{d^2} P(A_i|E_j)\left[ (d+1)P(E_j) - \frac{1}{d}\right].
\end{align}
\end{corollary}

\section{A Measure of Deformity}

As we've seen, the only difference between the law of total probability and the generalized Born rule is the presence of the Born matrix $\Phi$: compare
\begin{align}
P(A) = P(A|E)P(E) && \text{and} && P(A) = P(A|E)\Phi P(E).
\end{align}
Inspired by this, \cite{irreducible_difference} introduced the distance between the identity and $\Phi$,
\begin{align}
\mathcal{D}(\Phi) =||I - \Phi||	,
\end{align}
as a measure of this deformation, calling it the  \emph{quantumness}. The matrix norm is assumed to be unitarily invariant, e.g.\ a Schatten $p$-norm $||A||_p	= \left(\sum_i \sigma_i^p\right)^{\frac{1}{p}}$ where the $\sigma_i$'s are the singular values of $A$.
 By minimizing this quantity over all allowed reference devices with a certain number of outcomes, one obtains the \emph{irreducible quantumness}, which is a property of the underlying theory itself (e.g.\ $d$ dimensional quantum theory over $\mathbb{C}$), and which is meant to be a measure of nonclassicality on  grounds that classically one may always pick a reference measurement for which $\Phi=I$, simply reading off an objective property of the system and re-preparing it. The conclusion of \cite{irreducible_difference} was that parallel update SIC reference devices achieve the irreducible quantumness in all dimensions, at least if one considers only  MIC reference devices (with $d^2$ outcomes).
 
 We prefer here to call $\mathcal{D}(\Phi)$ the \emph{LTP deformation} since we are working in a broader context than quantum theory. If the states and effects of a reference device form a basis, as in the case of a MIC, then $\Phi = P(E|E)^{-1}$ univocally; but if we want to calculate $\mathcal{D}(\Phi)$ for ``overcomplete'' reference devices, we must make a choice of Born matrix. We introduce the \emph{minimal LTP deformation}
 \begin{align}
 	\mathcal{MD}(\Phi) = \min_{\Phi}||I - \Phi|| \quad \text{s.t.} \quad P\Phi P = P
 \end{align}
which is a constrained convex optimization problem with a unique solution, delivering the Born matrix which minimizes the LTP deformation for a choice of norm. Then, one can further obtain the \emph{irreducible LTP deformation} by minimizing over both reference devices with a certain number of outcomes as well as over $\Phi$ itself.

 For purposes of comparison, let us  define the \emph{natural Born matrix} for a depolarizing reference device. Just as before we defined an $\BT{E}$ such that $\B{S}\BT{E}=I$, we can similarly define an $\BT{S}$ such that $\BT{S}\B{E}=I$. Putting them together, yields
\begin{align}
\Phi &= \BT{E}\BT{S} = 
\left(\begin{array}{c|ccc}  
& & & \\
w & & \alpha\B{Z}^\dagger \\
& & &\\
\end{array}\right)\left(\begin{array}{ccc}  
1 & \dots & 1 \\
\hline
& \\
& \alpha\B{X}\\
& &\\
\end{array}\right)\\
&= wu^\dagger +\alpha^2 \B{Z}^\dagger \B{X}.
\end{align}
 The natural Born matrix indeed acts as an Protourgleichung:
\begin{align}
	\Phi P(E) &= \left[ wu^\dagger +\alpha^2 \B{Z}^\dagger \B{X}\right][w + \B{Z}^\dagger \rho_0]\\
	&= w+\alpha^2 \B{Z}^\dagger \B{X} \B{Z}^\dagger \rho_0\\
	&= w+\alpha \B{Z}^\dagger \rho_0.
\end{align}
Recall that $P=wu^\dagger +\B{Z}^\dagger \B{X}$. Since $\B{X}\B{Z}^\dagger = \frac{1}{\alpha}I$,  $\B{Z}^\dagger \B{X}$ and $\B{X}\B{Z}^\dagger$ have the same non-zero spectrum, and $\tr(wu^\dagger \B{Z}^\dagger \B{X})=0$, it follows that $P$ must have 1 eigenvalue equal to 1, $r-1$ eigenvalues equal to $\frac{1}{\alpha}$, and the rest  0. By the same argument, the natural Born matrix $\Phi$ will have 1 eigenvalue equal to 1, $r-1$ eigenvalues equal to $\alpha$, and the rest 0: in fact, this $\Phi$ forms a \emph{group inverse} of $P$, which is a \emph{spectral} inverse \cite{benisrael_greville, born_identity}.

In the overcomplete case, however, as we've observed, there are many alternative choices of $\Phi$. Upon picking a unitarily invariant norm, we can calculate the $\Phi$ which achieves the minimal LTP deformation $ \min_{\Phi}||I - \Phi|| \text{ s.t. } P\Phi P = P$. In fact, using the method of Lagrange multipliers, we can solve for the minimal $\Phi$ matrix with respect to the Schatten 2-norm analytically, at least for unbiased reference devices. 

\begin{theorem}
The unique Born matrix which minimizes $||I - \Phi||_2$ for an unbiased depolarizing reference device is
 	\begin{align}
 	\Phi =I + \frac{\alpha-1}{\alpha}\B{X}^\dagger (\B{X}\B{X}^\dagger)^{-1}(\B{Z}\B{Z}^\dagger)^{-1}\B{Z}.	
 	\end{align}
\end{theorem}
\begin{proof}
Define a Lagrangian
\begin{align}
\mathscr{L}(\Phi, \Lambda) &= ||I-\Phi||_2^2 + \tr\Big(\Lambda^T(P\Phi P-P)\Big)\\
&= n - 2\tr(\Phi) + ||\Phi||^2 + \tr(\Lambda^T P\Phi P) - \tr(\Lambda^T P)
\end{align}
where $\Lambda$ is a matrix of Lagrange multipliers.
Recalling that
\begin{align}
\frac{\partial \tr(X)}{\partial X} = I &&	\frac{\partial ||X||^2}{\partial X} = 2X && \frac{\partial \tr(AX)}{\partial X} = A^T, 
\end{align}
we can differentiate with respect to $\Phi$, and demand that the result vanishes.
\begin{align}
\frac{\partial\mathscr{L}(\Phi, \Lambda) }{\partial \Phi}	&= -2I+2\Phi + (P\Lambda^TP)^T=0
\end{align}
yields
\begin{align}
\Phi = I - \frac{1}{2}	P^T\Lambda P^T,
\end{align}
where, considering that $P\Phi P = P$, $\Lambda$ must satisfy
\begin{align}
\label{lambda}
P^2-P=\frac{1}{2}PP^T \Lambda P^TP.
\end{align}
Considering the LHS of Eq.\ (\ref{lambda}), and recalling that $\B{X}\B{Z}^\dagger = \B{Z}\B{X}^\dagger = \frac{1}{\alpha}I$ and $u^\dagger \B{Z}=\B{X}w=0$, we find
\begin{align}
P^2 &= 	\left[wu^\dagger + \B{Z}^\dagger \B{X}\right]\left[wu^\dagger + \B{Z}^\dagger \B{X}\right]= wu^\dagger + \frac{1}{a}\B{Z}^\dagger \B{X},
\end{align}
so that
\begin{align}
P^2 -P &=  wu^\dagger + \frac{1}{a}\B{Z}^\dagger \B{X}-wu^\dagger - \B{Z}^\dagger \B{X} = \frac{1-\alpha}{\alpha}\B{Z}^\dagger \B{X}.
\end{align}
As for the RHS of Eq.\ (\ref{lambda}), we first note that
\begin{align}
P P^T &= 	\left[ wu^\dagger + \B{Z}^\dagger \B{X}\right]\left[ uw^\dagger + \B{X}^\dagger \B{Z}\right]\\
&= w u^\dagger u w^\dagger + wu^\dagger \B{X}^\dagger \B{Z}+\B{Z}^\dagger \B{X}uw^\dagger + \B{Z}^\dagger \B{X}\B{X}^\dagger \B{Z}\\
P^TP &= 	\left[ uw^\dagger + \B{X}^\dagger \B{Z}\right]\left[ wu^\dagger + \B{Z}^\dagger \B{X}\right]\\
&= u w^\dagger w u^\dagger + uw^\dagger \B{Z}^\dagger \B{X}+\B{X}^\dagger \B{Z}wu^\dagger + \B{X}^\dagger \B{Z}\B{Z}^\dagger \B{X}.
\end{align}
 Let us assume that $\Lambda = \B{X}^\dagger \Lambda^\prime \B{Z}$, so that the terms in $PP^T$ which end in $w^\dagger$ and the terms in $P^T P$ which begin with $u$ all vanish, leaving
\begin{align}
	&=\frac{1}{2}\Big[ wu^\dagger \B{X}^\dagger \B{Z}+ \B{Z}^\dagger \B{X}\B{X}^\dagger \B{Z}\Big]\B{X}^\dagger \Lambda^\prime \B{Z}\Big[\B{X}^\dagger \B{Z}wu^\dagger+ \B{X}^\dagger \B{Z}\B{Z}^\dagger \B{X}\Big]\\
	&=\frac{1}{2\alpha^2}\Big[ wu^\dagger \B{X}^\dagger + \B{Z}^\dagger \B{X}\B{X}^\dagger \Big] \Lambda^\prime \Big[\B{Z}wu^\dagger+ \B{Z}\B{Z}^\dagger \B{X}\Big].
\end{align}
If we further assume that the measurement is unbiased, so that $w=\frac{1}{n}u$, then this becomes simply $\frac{1}{2\alpha^2}\Big[ \B{Z}^\dagger \B{X}\B{X}^\dagger \Big] \Lambda^\prime \Big[\B{Z}\B{Z}^\dagger \B{X}\Big]$.
Recall that informational completeness implies that $\B{XX}^\dagger$ and $\B{ZZ}^\dagger$ are invertible. So let $\Lambda^\prime = \beta(\B{X}\B{X}^\dagger)^{-1} (\B{Z}\B{Z}^\dagger)^{-1}$. As for the constant $\beta$, we require
\begin{align}
 \frac{1}{2\alpha^2} \Big[\B{Z}^\dagger \B{X} \B{X}^\dagger \Big]\beta (\B{X}\B{X}^\dagger)^{-1} (\B{Z}\B{Z}^\dagger)^{-1}\Big[ \B{Z}\B{Z}^\dagger \B{X}\Big]=\frac{\beta}{2\alpha^2}\B{Z}^\dagger \B{X} \overset{!}{=}\frac{1-\alpha}{\alpha}\B{Z}^\dagger \B{X}.
\end{align}
Thus $\beta = 2\alpha(1-\alpha)$, and
\begin{align} 
\Lambda &= 2\alpha(1-\alpha)\B{X}^\dagger (\B{X}\B{X}^\dagger)^{-1}(\B{Z}\B{Z}^\dagger)^{-1}\B{Z}.
\end{align}
Putting it all together, the unique Born matrix which minimizes the LTP deformation with respect to the Frobenius norm for an unbiased depolarizing reference device is
\begin{small}
\begin{align}
\Phi &= I - \frac{1}{2}	P^T\Lambda P^T\\
&= I - \frac{1}{2}\left[ \frac{1}{n}uu^\dagger + \B{X}^\dagger \B{Z}\right]\left[2\alpha(1-\alpha)\B{X}^\dagger (\B{X}\B{X}^\dagger)^{-1}(\B{Z}\B{Z}^\dagger)^{-1}\B{Z}\right]\left[ \frac{1}{n}uu^\dagger + \B{X}^\dagger \B{Z}\right]\\
&= I+\alpha(\alpha-1)\B{X}^\dagger \B{Z}\Big[ \B{X}^\dagger (\B{X}\B{X}^\dagger)^{-1}(\B{Z}\B{Z}^\dagger)^{-1}\B{Z}\Big] \B{X}^\dagger \B{Z}\\
&= I + \frac{\alpha-1}{\alpha}\B{X}^\dagger (\B{X}\B{X}^\dagger)^{-1}(\B{Z}\B{Z}^\dagger)^{-1}\B{Z}.
\end{align}
\end{small}
Uniqueness follows from the fact that the optimization problem is convex.
\end{proof}

But how does this Born matrix act on a probability vector? Recalling that $P(E)=\B{E}|\rho)=w + \B{Z}^\dagger \rho_0 = \frac{1}{n} u + \B{Z}^\dagger \rho_0$ in the unbiased case, we have
\begin{align}
\Phi P(E) &= \Big[ I + \frac{\alpha-1}{\alpha}\B{X}^\dagger (\B{X}\B{X}^\dagger)^{-1}(\B{Z}\B{Z}^\dagger)^{-1}\B{Z}\Big]\Big[	\frac{1}{n}u + \B{Z}^\dagger \rho_0\Big]\\
&= \frac{1}{n} u + \B{Z}^\dagger \rho_0 + \frac{\alpha-1}{\alpha}\B{X}^\dagger (\B{X}\B{X}^\dagger)^{-1}(\B{Z}\B{Z}^\dagger)^{-1}\B{Z}\B{Z}^\dagger \rho_0\\
&= \frac{1}{n} u + \Big(\B{Z}^\dagger + \frac{\alpha-1}{\alpha}\B{X}^\dagger (\B{X}\B{X}^\dagger)^{-1}\Big)\rho_0.
\end{align}
In other words, this Born matrix does \emph{not} act as an Protourgleichung! 

That said, there is a special case. We can always choose $\B{X} = \frac{1}{\alpha}(\B{Z}^\dagger)^+=\frac{1}{a}(\B{Z}\B{Z}^\dagger)^{-1}\B{Z}$---that is, $\B{X}$ can be taken proportional to the pseudoinverse of $\B{Z}^\dagger$, which is a left inverse since $\B{Z}^\dagger$ has linearly independent columns. Symmetrically, we have $\B{Z}^\dagger=\frac{1}{\alpha}\B{X}^\dagger(\B{X}\B{X}^\dagger)^{-1}$. Thus
\begin{align}
\Phi &= I + \alpha(\alpha-1)\B{Z}^\dagger \B{X}.
\end{align}
and
\begin{align}
\Phi P(E) &= \frac{1}{n} u + \Big(\B{Z}^\dagger + (\alpha-1)\B{Z}^\dagger \Big)\rho_0= \frac{1}{n} u + \alpha\B{Z}^\dagger\rho_0.
\end{align}
In this special case, the minimal $\Phi$ acts as a Protourgleichung. In fact, even if the measurement is biased and the states were not obtained from the pseudoinverse, we can  nevertheless always use a Born matrix of the form $\Phi =  I + \alpha(\alpha-1)\B{Z}^\dagger \B{X}$---we'll call such a Born matrix \emph{simple}. Indeed,
\begin{align}
P\Phi P &= \left[ wu^\dagger + \B{Z}^\dagger \B{X}\right] \left[I + \alpha(\alpha-1)\B{Z}^\dagger \B{X} \right] \left[ wu^\dagger + \B{Z}^\dagger \B{X}\right] \\
&=wu^\dagger +\frac{1}{\alpha}\B{Z}^\dagger \B{X}+\alpha(\alpha-1)\left(\frac{1}{\alpha^2}\right)\B{Z}^\dagger \B{X}\\
&= wu^\dagger +\B{Z}^\dagger \B{X} = P.
\end{align}

So what have we learned? Recall that for a MIC, we have
\begin{align}
 P(E|E)&=\frac{1}{\alpha}I + \left(1-\frac{1}{\alpha}\right)wu^\dagger,
\end{align}
and it is not hard to confirm that the unique $\Phi$ matrix is
\begin{align}
\Phi &= \alpha I + (1-\alpha)wu^\dagger,	
\end{align}
which clearly acts as an Urgleichung. For overcomplete reference devices, however, we must make a choice of $\Phi$. On the one hand, any depolarizing reference device has at least two Born matrices which act as a Protourgleichung,
\begin{align}
\Phi =wu^\dagger +\alpha^2 \B{Z}^\dagger \B{X} && \Phi = 	I + \alpha(\alpha-1)\B{Z}^\dagger \B{X},
\end{align}
which clearly yield different values for the LTP deformation $||I-\Phi||$: the measure is, in this sense, over-sensitive, distinguishing the two matrices despite the fact that they yield the same quasi-probabilities. On the other hand, the Born matrix which minimizes the LTP deformation with respect to the Frobenius norm
\begin{align}
\Phi =I + \frac{\alpha-1}{\alpha}\B{X}^\dagger (\B{X}\B{X}^\dagger)^{-1}(\B{Z}\B{Z}^\dagger)^{-1}\B{Z}
\end{align}
does not in general act as a Protourgleichung.   This arguably represents a flaw in the metric itself, if the original purpose was to quantify the \emph{simplicity} of the LTP deformation. Indeed, these examples show that the two concepts (minimal deformation with respect to a matrix norm, and the form of the Urgleichung itself) are in fact distinct. On the other hand, the exercise has allowed us to identify a distinguished set of reference states for any unbiased informationally complete reference measurement, namely those given by
\begin{align}
\B{S} =\left(\begin{array}{ccc}  
1 & \dots & 1 \\
\hline
& \\
& \frac{1}{a}(\B{Z}\B{Z}^\dagger)^{-1}\B{Z} \\
& &\\
\end{array}\right)
\end{align}
for which the minimal $\Phi$ (with respect to the 2-norm) does in fact act as a Protourgleichung. In fact, if the measurement is weight-morphophoric, then this just gives us a parallel update reference device.
\section{Conclusion}

In this work, we've established that the Protourgleichung 
\begin{align}
P(A_i) = \sum_{j=1}^{d^2} P(A_i|E_j)\Big[ \alpha P(E_j) + (1-\alpha)w_j\Big]
\end{align}
is ubiquitous. In any GPT, and for any reference measurement, one can always find a set of reference states which yield an expression of this form. The key requirement is that the reference device, considered as a measure-and-prepare channel, must be depolarizing---and this requirement can always be satisfied. Considering the case where the reference states can be chosen proportional to the effects (modulo a depolarization parameter), we related our construction to so-called ``morphophoric measurements,'' and ultimately to complex projective 2-designs, from which we recovered the standard Urgleichung. In light of these results, we reconsidered a matrix-norm-based measure of LTP deformity, showing by example that it fails to adequately capture the essence of the Protourgleichung, even as it still holds interesting lessons about the nature of reference devices to teach us.

The Protourgleichung is the simplest normative constraint between probability assignments $P(A_i)$, $P(A_i|E_j)$, and $P(E_j)$ which is compatible with the assumptions of generalized probabilistic theories: it is universally available whether one works in quantum theory, classical probability theory, classical theory with an epistemic restriction, or any of the more exotic members of the GPT community for which, indeed, no noncontextual ontological model is possible. For a broader look at the philosophy behind nonclassical reference devices themselves, consider \cite{born_identity}. Here we observe that it is  illuminating to rewrite the Protourgleichung as 
\begin{align}
P(A)_\rho = P(A|E)\Big[\alpha P(E)_\rho+(1-\alpha)P(E)_M\Big],
\end{align}
where $P(E)_M = w$, the reference outcome probabilities given the maximally mixed state, which we can rearrange as 
\begin{align}
P(A)_\rho - P(A|E)P(E)_M&=\alpha\Big[ P(A|E)P(E)_\rho - P(A|E)P(E)_M\Big]	\\
P(A)^{\text{(Born)}}_\rho - P(A)_M^{\text{(LTP)}}&= \alpha\Big[ P(A)^{\text{(LTP)}}_\rho - P(A)^{\text{(LTP)}}_M\Big].	
\end{align}
In other words, after subtracting off the probabilities $P(A)^{\text{(LTP)}}_M$, which depend on the bias of the measurement, the Protourgleichung's predictions and the law of total probability's predictions are actually proportional to each other: nothing simpler or more elegant could be imagined. 

Nevertheless, the Urgleichung itself is a child of quantum theory over $\mathbb{C}$, its particular parameter ($\alpha=d+1$) and the cardinality of its outcomes ($d^2$) hinting at quantum theory's special structure, and whispering of the possibility of \emph{pure} SIC reference devices. Indeed, these latter arguably represent the simplest nonclassical generalization of the trivial classical reference device, which simply reads off an objective property of a system, and re-prepares a state with exactly that same property. 

Consider that, in the terms we've developed here, the trivial classical reference device has states and effects which are simply the $r$-dimensional basis vectors. It is a parallel update device, with reference states proportional to reference effects. Considered as a measure-and-prepare channel, it is simply the identity channel, which is trivially depolarizing, and the Protourgleichung collapses to the LTP. The measurement is trivially morphophoric, simply mapping the probability simplex exactly to itself. It is unbiased. It is a MIC. It is equiangular. Its reference states are extremal, corresponding to the vertices of the probability simplex itself. The reference effects are extremal too. Each reference state can be perfectly distinguished by some effect\footnote{That is, there is some effect which yields probability 1 on that state.}, and each reference effect can be perfectly distinguished by some state. Moreover, the reference states and reference effects perfectly distinguish each other. Finally, such a device exists in any dimension, where the maximal number of pairwise perfectly distinguishable states scales with the dimension.

Compare this to the properties of a SIC reference device in quantum theory over $\mathbb{C}$, which is a parallel update device, which, considered as a measure-and-prepare channel, is depolarizing with parameter $\frac{1}{\alpha} = \frac{1}{d+1}$. It is morphophoric, unbiased, a MIC with $d^2$ elements. It is equiangular. Its reference states and reference effects are extremal, corresponding to the vertices of a simplex. Each reference state can be perfectly distinguished by some effect, and each reference effect can be perfectly distinguished by some state. And finally, such a device (apparently) exists in any dimension, where the maximal number of pairwise perfectly distinguishable states scales with the dimension\footnote{The point of emphasizing that the maximal number of pairwise perfectly distinguishable states scales with the dimension is to exclude, for example, so-called \emph{spin-factor} theories, whose state spaces are $n$-dimensional balls, and for which SICs exists trivially. In such theories, regardless of the dimension, the maximally number of pairwise perfectly distinguishable states is always 2. Such a theory is an example of a \emph{Euclidean Jordan Algebra} \cite{stacey2023quantum, wilce2017royal} of which quantum theory over $\mathbb{C}$ is also an example, as well as quantum theory over $\mathbb{R}, \mathbb{H}$, and $d=3$ quantum theory over $\mathbb{O}$. For quantum theory over $\mathbb{R}$, pure SICs are only known to exist in $d=2,3,7,23$, and are known \emph{not} to exist in most dimensions. For quantum theory over $\mathbb{H}$, pure SICs are only known to exist in $d=2,3$, and there is numerical evidence suggesting that they don't exist already in $d=4$ \cite{Cohn_2016}. There is a pure SIC in octonionic $d=3$, but this is the only dimension in which such a theory can be constructed.}.

The only differences, then, are the values of the depolarizing parameter, the fact that the state space is mapped to a subset of the probability simplex, and the fact that the reference states and reference effects don't perfectly distinguish each other. Thus pure SIC reference devices are arguably the closest analogue to the trivial classical reference device, and this is true not just in one theory, but in a whole infinite sequence of them, for each Hilbert space dimension $d$. Indeed, if one were trying to generalize classical probability theory as minimally as possible, incorporating the idea that nature forgoes hidden variables\footnote{Or more positively, as Blake Stacey would put it, respecting nature's \emph{vitality} \cite{stacey2023quantum}.}, one could imagine nothing simpler than positing the existence of pure SIC reference devices, with their attendant, the Urgleichung.

\section*{Acknowledgements}

This work was supported in part by National Science Foundation Grant 2210495 and in part by Grant 62424 from the John Templeton Foundation. The opinions expressed in this publication are those of the author and do not necessarily reflect the views of the John Templeton Foundation.

\section*{Appendix: Real Vector Space QM}

A close cousin of standard quantum theory is quantum mechanics over $\mathbb{R}$. Here there is a subtlety if one works with vectorized operators. For quantum mechanics over $\mathbb{C}$, if we take $P(E|S) = \B{ES}$ where the rows of $\B{E}$ are vectorized effects $(E_i|$ and the columns of $\B{S}$ are vectorized states $|S_i)$, then $\B{ES}$ will be a \emph{full rank factorization} \cite{full_rank_factorization}: this is because the real space of Hermitian matrices is $d^2$, but also the vectorized operators are $d^2$ dimensional. On the other hand, if one works over $\mathbb{R}$, the state space is the space of real symmetric matrices which has dimension $d(d+1)/2$ even as the vectorized operators remain $d^2$ dimensional. In this latter case, $\B{ES}$ built from vectorized operators will \emph{not} be a full rank factorization. One could instead work with the states and effects expanded in an operator basis of $d(d+1)/2$ elements; otherwise, one will have to modify the Born identity to $\bold{S}\Phi \bold{E} = \Pi_{\text{sym}^2}$ where $\Pi_{\text{sym}^2}$ is the projector onto symmetric subspace on two factors, i.e.\ onto the space of vectorized symmetric matrices. This issue is quite general if one doesn't work with full rank factorizations.

Thus for a depolarizing channel in real vector space quantum mechanics, we have (in terms of vectorized operators)
\begin{align}
	\sum_{i=1}^n |S_i)(E_i|&=	 \frac{1}{\alpha}\Pi_{\text{sym}^2} + \left(1-\frac{1}{\alpha}\right)\frac{1}{d}|I_d)(I_d|.
\end{align}
This theory is also self-dual, so  taking the trace of both sides, and assuming our states are pure, we find
\begin{align}
	d =	 \frac{1}{\alpha}\frac{d(d+1)}{2} + 1-\frac{1}{\alpha} \Longrightarrow \alpha = \frac{d+2}{2},
\end{align}
and
\begin{align}
	\sum_{i=1}^n |S_i)(E_i|&=	 \frac{2}{d+2}\Pi_{\text{sym}^2} + \frac{1}{d+2}|I_d)(I_d|\\
	&= \frac{1}{d+2}\Big(I_{d^2} +\text{SWAP} +|I_d)(I_d|\Big).
\end{align}
Again dividing by $\frac{1}{d}$ so the trace of both sides is 1, and taking the partial transpose, which actually leaves both sides of the equation invariant, we obtain
\begin{align}
	 \sum_{i=1}^n w_i S_i \otimes S_i =\frac{1}{d(d+2)}\Big(I_{d^2} +\text{SWAP} +|I_d)(I_d|\Big),
\end{align}
which according to \cite{harrow} coincides precisely with $\int |\psi\rangle\langle\psi|^{\otimes 2} d\psi$ where the integral is taken over pure states in the real vector space theory. Thus (pure) measure-and-prepare weight-morphophoric references devices for quantum theory over $\mathbb{R}$ coincide with real projective 2-designs.

On the subject of real vector space quantum mechanics, \cite{weiss} conducted a numerical search for MIC reference devices minimizing $||I - \Phi||$ in $d=4$. It is known that in this dimension one cannot construct a SIC-POVM out of pure effects: the maximum of 6 equiangular lines is not enough to furnish the 10 elements required for a MIC. In fact, in $d=4$, the minimum number of elements in a pure biased 2-design is 11, and the minimum number of elements in a pure unbiased 2-design is 12 \cite{HUGHES202184}. These efforts were unable to yield minimal reference devices whose Born matrix acts as an Protourgleichung, a fact which is now contextualized by the present work. Indeed, such reference devices \emph{can} be found, although they will not minimize $||I - \Phi||$. It is worth noting that subsequent numerical searches suggest that in this case, if the reference effects are pure, then the depolarizing reference states must be mixed. It is also worth noting that unlike in \cite{irreducible_difference}, the minimal reference device does depend on the choice of matrix norm, which is likely to be true in the general overcomplete case.

\printbibliography

\end{document}